\documentclass[letterpaper, 10 pt, conference]{ieeeconf}  

\IEEEoverridecommandlockouts                              

\usepackage{bm}
\usepackage{multicol}
\usepackage{amsmath}
\usepackage{graphicx}
\usepackage{amssymb}
\usepackage{float}
\usepackage{mathrsfs}
\usepackage{comment}
\usepackage{mathtools}
\usepackage{yfonts}
\usepackage{caption}
\usepackage{subcaption}
\usepackage{float}
\usepackage{algorithm}
\usepackage{algpseudocode}
\usepackage{xcolor}
\usepackage{commands}
\usepackage{hyperref}

\hypersetup{
    colorlinks=true,
  citecolor=blue,
  linkcolor=red,
  urlcolor=magenta
}
 
\title{\LARGE \bf
RAPID: Autonomous Multi-Agent Racing using Constrained  Potential Dynamic Games
}

\author{Yixuan Jia$^{1}$, Maulik Bhatt$^{2}$, and Negar Mehr$^{2}$
\thanks{This work is supported by the National Science Foundation, under grants ECCS-2145134 CAREER Award, CNS-2218759, and CCF-2211542.}
\thanks{$^{1}$Yixuan Jia is with the Department of Electrical and Computer Engineering, University of Illinois Urbana-Champaign,
        306 N Wright St, Urbana, IL 61801, USA
        {\tt\small yixuanj3@illinois.edu}}%
\thanks{$^{2}$Maulik Bhatt and Negar Mehr are with the Department of Aerospace Engineering, University of Illinois Urbana-Champaign,
        104 S Wright St, Urbana, IL 61801, USA
        {\tt\small  mcbhatt2@illinois.edu, negar@illinois.edu }}%
}

\begin{document}

\maketitle
\thispagestyle{empty}
\pagestyle{empty}

\begin{abstract}
In this work, we consider the problem of autonomous racing with multiple agents where agents must interact closely and influence each other to compete. We model interactions among agents through a game-theoretical framework and propose an efficient algorithm for tractably solving the resulting game in real time. More specifically, we capture interactions among multiple agents through a constrained dynamic game. We show that the resulting dynamic game is an instance of a simple-to-analyze class of games. Namely, we show that our racing game is an instance of a constrained dynamic potential game. An important and appealing property of dynamic potential games is that a generalized Nash equilibrium of the underlying game can be computed by solving a single constrained optimal control problem instead of multiple coupled constrained optimal control problems. Leveraging this property, we show that the problem of autonomous racing is greatly simplified and develop RAPID (autonomous multi-agent RAcing using constrained PotentIal Dynamic games), a racing algorithm that can be solved tractably in real-time. Through simulation studies, we demonstrate that our algorithm outperforms the state-of-the-art approach. We further show the real-time capabilities of our algorithm in hardware experiments. 
\end{abstract}

\section{Introduction}
Autonomous racing is gaining popularity because of its broad applicability in various competitive and non-cooperative motion planning scenarios. Multi-agent autonomous racing is a highly challenging motion planning task, requiring multiple nonlinear agents to plan their motions in real time while operating at their limits. Furthermore, they must account for other agents with conflicting objectives and ensure safety constraints, such as avoiding collisions and staying on the track. This results in a set of coupled motion planning problems that are highly nonlinear and complex. Such complexities require efficient motion planning algorithms for real-time capabilities and to ensure safety while generating competitive trajectories.

One of the earlier approaches to autonomous racing in the context of RC cars was \cite{liniger2015optimization}, where the authors employed an optimization-based model predictive controller to maximize progress on the track, subject to safety requirements. However, due to the reactive nature of such approaches, they do not generate competitive trajectories. Learning-based approaches for autonomous racing were studied in \cite{rosolia2017autonomous, balaji2020deepracer, song2021autonomous, herman2021learn}. However, these methods mainly focus on the path-planning aspect of racing instead of interactions among the agents and may fail to generate competitive behaviors such as blocking and overtaking.

\begin{figure}[t]
    \centering
    \includegraphics[scale=0.36]{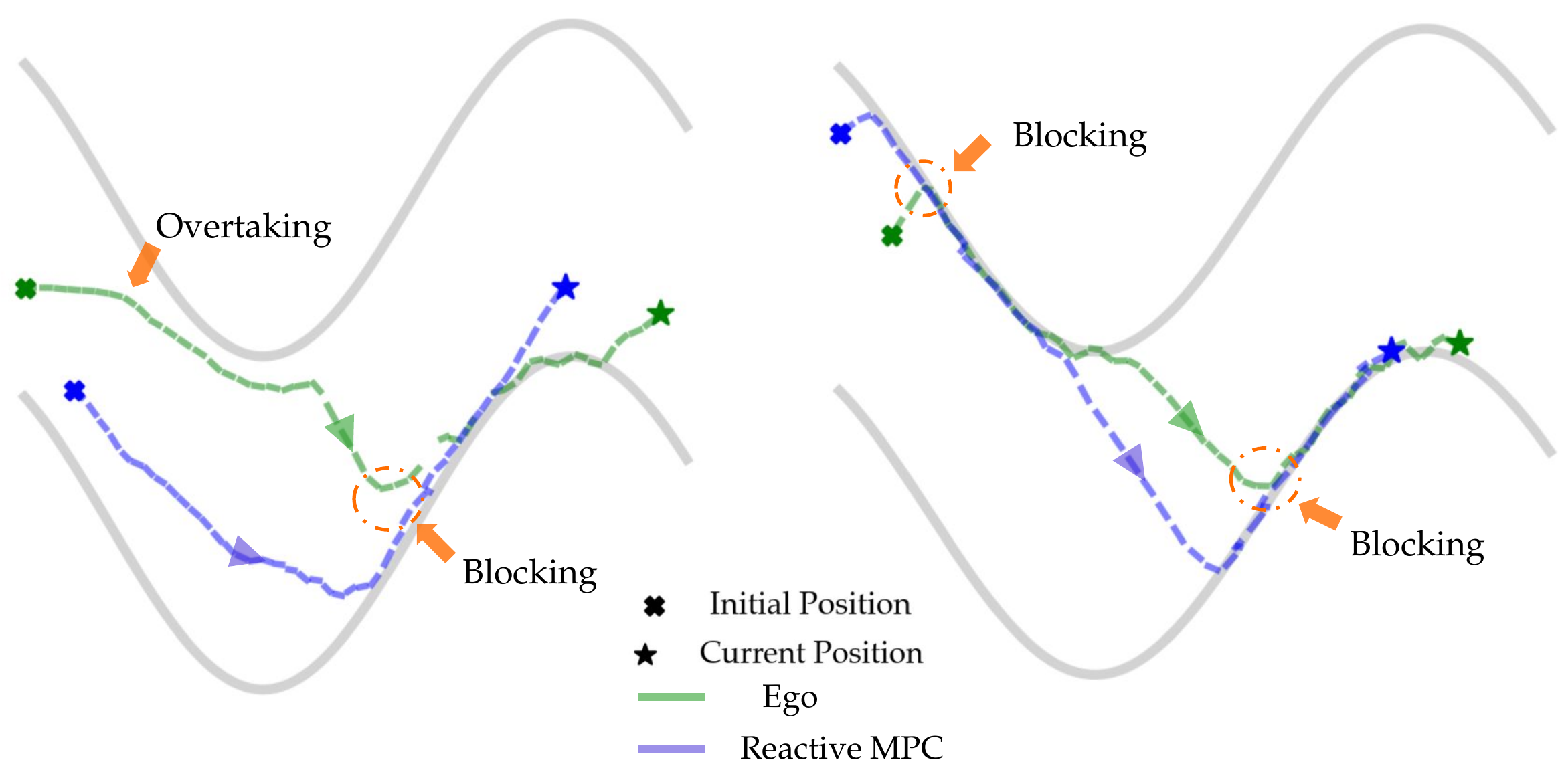}
    \caption{Visualizations of two of the hardware experiments using two Crazyflies. 
    As it can be seen in the left figure, even when the ego drone starts behind, it overtakes the opponent and ultimately blocks the opponent to avoid being overtaken. On the right, it can be seen that when the ego drone starts ahead of the opponent, it generates blocking behavior twice to avoid being overtaken.
    No collisions occurred in the hardware experiments and the ego-drone won the race in both instances.}
    \label{fig:two_player_hardware_experiment}
    \vspace{-0.5cm}
\end{figure}
Due to the interactive nature of racing, where each agent has to account for other agents' decisions, interactions among agents can naturally be captured in a game-theoretic framework. Game-theoretic planning has been extensively used in non-cooperative motion planning for multiple agents \cite{sadigh2016planning, dreves2018generalized, fridovich2020efficient, zanardi2021urban,wang2020game, mehr2023maximum}. 
Due to its success in motion planning, game theory has also recently been applied in autonomous racing. 
Non-cooperative game theoretic planners for autonomous racing involving two agents were studied in \cite{8643396, 9112709}. 
However, these works apply to only two-agent settings. 
Going beyond two agents, it is often computationally challenging to account for interactions among multiple agents due to the nonlinearities inherent in racing. 

In \cite{9329208}, a game-theoretic planner was proposed for racing among multiple agents using sensitivity-based analysis. However, the proposed algorithm does not scale well with the number of agents. A game-theoretic planner with data-driven identification of vehicle models was developed for head-to-head autonomous racing in \cite{jung2021game}. This work uses Stackelberg strategies which are not generally suitable for multi-agent racing as they normally involve a leader-follower structure.

\begin{figure*}[t]
    \centering
    \includegraphics[scale = 0.7]{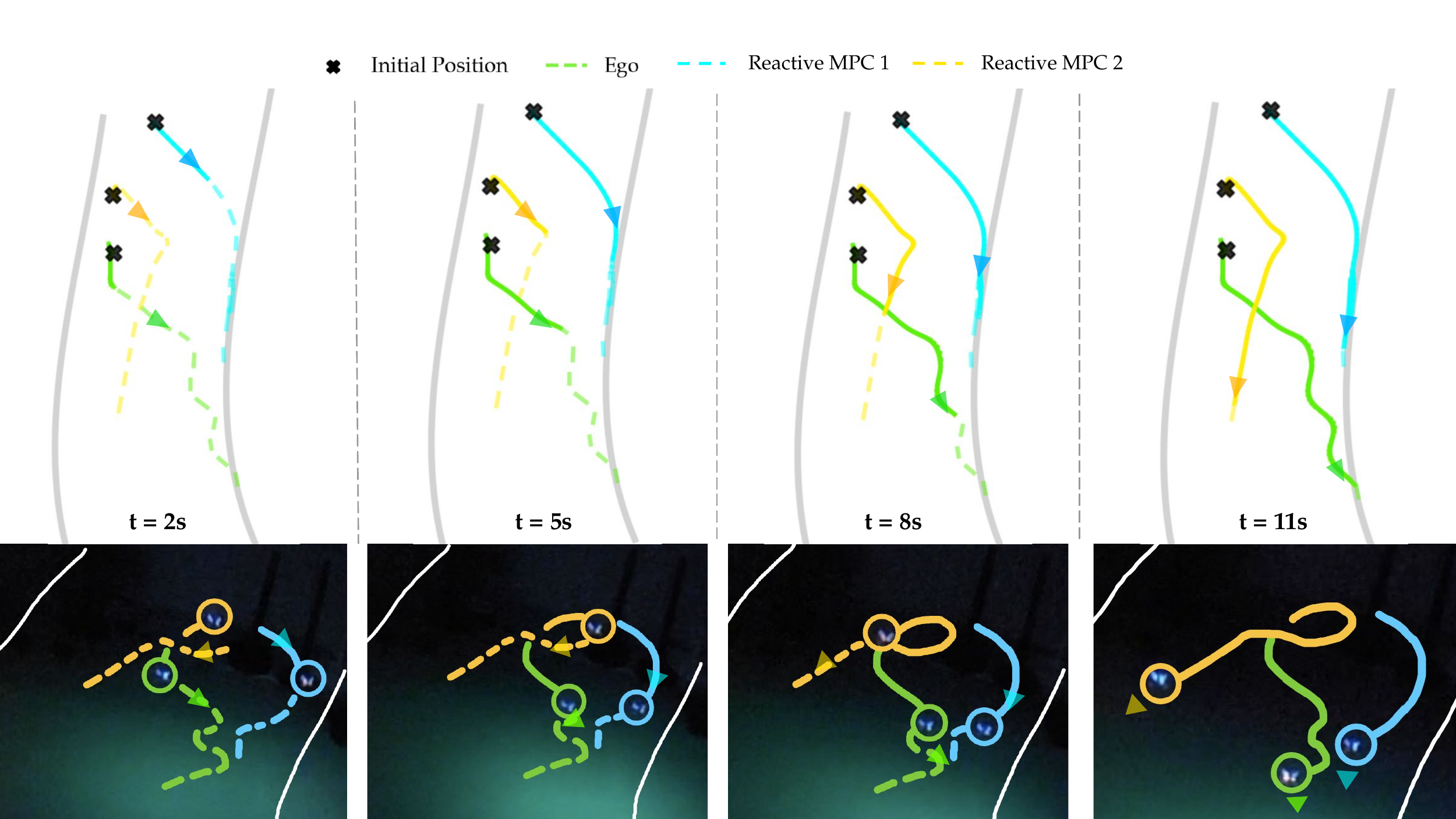}
    \caption{Experiment snapshots of the race among three quadrotors. Green color is used to represent the ego quadrotor which uses our trajectory planner RAPID and exhibits an overtaking maneuver around the two quadrotors that use reactive MPC. The top row consists of top views that are plotted with the actual hardware experiment data. The second row shows the snapshots of the actual hardware experiment. For visualization purposes, the approximated boundaries of the track are marked by the white lines on the second row. We see that the ego (green) is trying to block the agent marked with the blue color. Videos can be found at: \url{https://youtu.be/85PYCj6vUd4}.}
    \label{fig:three_player_experiment}
    \vspace{-0.5cm}
\end{figure*}

In this work, we present RAPID, an autonomous multi-agent racing algorithm that uses constrained dynamic potential games. We pose the racing problem as a non-cooperative constrained general-sum dynamic game and seek the Nash equilibria of the game. We show that although, in general, finding generalized Nash equilibria --- Nash equilibria of constrained dynamic games--- is very challenging, the equilibria of our racing game can be found tractably and efficiently. 
Our key insight is that our racing game is an instance of a \emph{dynamic potential game} for which equilibria always exist and can be found by solving a single constrained optimal control problem. The advantage of formulating the problem as a dynamic potential game is that it is generally more efficient to solve the underlying multivariate optimal control problem than to solve a set of coupled optimal control problems \cite{7448954, gonzalez2013discrete}.
We leverage this property and develop a tractable planning algorithm for racing among multiple agents. We compare our method with the state-of-the-art and demonstrate that our method beats the existing work in terms of both the computation time and the quality of the generated trajectories. We further demonstrate the real-time capabilities of our method in a hardware racing experiment involving three quadcopters.


\section{Problem Formulation}
\vspace{-0.15cm}
\subsection{Notations}
We consider the problem of autonomous racing in discrete time. Let $T \in \mathbb{N}$ denote the number of time steps in the problem, where $\mathbb{N}$ denotes the set of natural numbers. For any natural number $n$, let $\left[n\right] \coloneqq \{1, \ldots, n\}$ be the set of all natural numbers smaller than or equal to $n$. Let $\left[N\right]$, where $N \in \mathbb{N}$, denote the set of agents' indices, where each index corresponds to an agent. For each agent $i \in \left[N\right]$, agent $i$'s state is denoted as $x^i \in \mathcal{X}^i \subset \mathbb{R}^{n_i}$ where $\mathcal{X}^i$ is the state space for agent $i$ and $n_i \in \mathbb{N}$ is the dimension of state space of agent $i$. Similarly, the control action for each agent $i \in \left[N\right]$ is denoted as $u^i \in \mathcal{U}^i \subset \mathbb{R}^{m_i}$, where $\mathcal{U}^i$ is the action space for agent $i$ and $m_i\in \mathbb{N}$ is the dimension of the control space of agent $i$. 

Let $\mathcal{X} \coloneqq \mathcal{X}^1 \times \ldots \times \mathcal{X}^N$ denote the entire state space of the system and let $\mathcal{U} \coloneqq \mathcal{U}^1 \times \ldots \times \mathcal{U}^N$ be the action space of all agents participating in the race. 
We denote the state of all agents at time step $t$ as $x_t \coloneqq (x^{1^{\top}}_t, \ldots, x^{N^{\top}}_t )^\top \in \mathcal{X} \subset \mathbb{R}^{n}$, and the control action of all agents at time $t$ can be denoted as $u_t \coloneqq (u^{1^{\top}}_t, \ldots, u^{N^{\top}}_t)^{\top} \in \mathcal{U} \subset \mathbb{R}^{m}$, where $n \coloneqq \sum_{i \in [N]}n_i$ and $m \coloneqq  \sum_{i \in [N]}m_i$. Furthermore, we use $u^{-i}_t := (u^{1^{\top}}_t,\ldots,u^{{i-1}^{\top}}_t,u^{{i+1}^{\top}}_t,\ldots,u^{N^{\top}}_t)^{\top} \in \mathcal{U}^{-i}$, where $\mathcal{U}^{-i} := \Pi_{j\neq i}\mathcal{U}^j$, to denote the control actions of all agents at time step $t$ except for agent $i$. With a slight abuse of
notations, we can write $u_t = (u^{i^{\top}}_t,u^{{-i}^{\top}}_t)^\top$.


We assume that for each agent $i \in \left[N\right]$, its state at every time step $x_t^i$ evolves according to the discrete dynamics $x^i_{t+1} = f^{i}(x^i_t, u^i_t, t)$, where $f^i: \mathcal{X}^i\times  \mathcal{U}^i \times \left[N\right] \to \mathbb{R}^{n_i}$, and $x^i_t, ~u^i_t$ denote the state and action of agent $i$ at time step $t$.
We define the overall system dynamics as $f \coloneqq \left(f^1, \ldots, f^N\right)$. Then, the evolution of the joint state of the system is described by:
\vspace{-0.15cm}
\begin{equation}\label{eq:dynamics}
    x_{t+1} = f(x_t, u_t, t).
\end{equation}
 
While racing against each other, each agent has to satisfy some constraints, such as staying on the race track and avoiding collisions with obstacles and other agents. We denote the constraints of the race track through the function,
\begin{equation}\label{eq:track_constraint}
    h(x^i_t) \leq 0, \forall i \in [N], \forall t \in \{0,1,\ldots,T-1\}.
\end{equation}
We consider collision avoidance constraints which are defined as follows for every pair of agents:
\vspace{-0.15cm}
\begin{equation}\label{eq:collision_constraint}
    d(x^i_t,x^j_t) > d_{min}, \quad \forall i < j, i, j \in [N],
\end{equation}
where $d(x^i_t,x^j_t)$ is the Euclidean distance between agents $i$ and $j$ at time step $t$ and $d_{min}$ is the minimum distance between before which two agents are considered collided.
We also consider upper and lower bounds on the control inputs of each agent. For each agent $i \in [N]$, we have
\begin{equation}\label{eq:control_constraint}
    u^i_{min} \leq u^i_t \leq u^i_{max}, \forall t \in \{0,1,\ldots,T-1\}
\end{equation}
where $u^i_{min}$ and $u^i_{max}$ are the vectors of minimum and maximum control inputs for agent $i$, respectively. 
Let $g^i(x, u, t) \leq 0$ denote the concatenated vector of constraints in \eqref{eq:track_constraint}, \eqref{eq:collision_constraint} and \eqref{eq:control_constraint} for agent $i$, where $g^i: \mathcal{X}\times \mathcal{U} \times \left[N\right] \to \mathbb{R}^{c_i}$ is a vector-valued function and $c_i \in \mathbb{N}$ is the number of constraints of agent $i$. Note that the inequality is to be understood element-wise. It should also be noted that the constraint function of agent $i$, in general, depends on the states and control actions of \emph{all agents}. This results in a coupling between agents' decisions since one agent cannot satisfy their constraint function in isolation. Hence, agents must account for such couplings while choosing their control actions. We can collect constraints for all agents and define $g \coloneqq (g^{1^\top}, \ldots, g^{N^\top})^\top : \mathcal{X}\times \mathcal{U} \times \left[N\right] \to \mathbb{R}^{c}$, where $c \coloneqq \sum_{i \in [N]} c_i$.  Then the combined constraint function for the entire system is:
\begin{align}\label{eq:constraints}
 g(x_t, u_t, t) \leq 0.
\end{align}

We define $\mathcal{C}_t$ to be the set of all feasible states and actions at time $t$, i.e., the constrained subset of $\mathcal{X}\times\mathcal{U}$ which satisfies all the constraints~\eqref{eq:constraints} at time step $t$. 
We can define $\mathcal{C}_0 = \mathcal{X}\times\mathcal{U} \cap \{(x_0,u_0): g(x_0,u_0)\leq 0\}$ and $\mathcal{C}_t = \{\{\mathcal{X} \cap \{x_t: x_t = f(x_{t-1},u_{t-1},t-1)\} \}\times\mathcal{U}\}\cap \{(x_t,u_t): g(x_t,u_t)\leq 0\}$ for $t\in\{1,\ldots,T-1\}$. At the terminal time step $T$, no action is being taken, and; hence, the constraint function will only be a function of states $\left(g(x_{T},(\cdot),T)\right)$ and the corresponding constraint set will be $\mathcal{C}_{T} = \mathcal{X} \cap \{x_{T}: g(x_{T},(\cdot),T) \leq 0 \}$. 

\subsection{Agents' Strategies and Objectives}

Each agent wants to choose actions sequentially to maximize their chances of winning the race. The policy through which each agent chooses its actions is called its strategy. A strategy can have various forms. In the present scenario, we consider open-loop strategies for the agents, which means that the strategy of each agent depends only on the initial state of the system and time.\footnote{We acknowledge that considering open-loop strategies may generate different trajectories in contrast to choosing closed-loop strategies which can be a function of the system's state at every time step. However, we implement the open-loop strategies in a receding horizon fashion to account for the new information at each time step and mimic feedback strategies. We show empirically that this is a reasonable approximation for practical purposes in motion planning}. Let $\mathcal{T}= \{0,1,\ldots,T\}$. For any agent $i \in [N]$, we define the open-loop {strategy} of agent $i$, $\gamma^i : \mathcal{X} \times \mathcal{T} \to \mathcal{U}^i $ as follows:
\begin{align*}
\gamma^i(x_0, t) := u^i_t.
\end{align*}
In other words, at each time step $t$, given the initial state of the system $x_0$, $\gamma^i(x_0, t)$ would be the control action chosen by agent $i$. We use $\Gamma^i$ to denote the space of all possible strategies for agent $i$. We define $\gamma \coloneqq ({\gamma^1}^\top, \ldots, {\gamma^N}^\top)^\top \in \Gamma$,
where $\Gamma \coloneqq \Gamma^1 \times \ldots \times \Gamma^N$
is the combined strategy space of the system. Let $\Gamma^{-i}:= \Pi_{j\neq i}\Gamma^j$ denote the strategy space of all agents except agent $i$. Same as before, we use $\gamma^{-i}:=({\gamma^1}^\top,\dots,{\gamma^{i-1}}^\top,{\gamma^{i+1}}^\top,\ldots,{\gamma^N}^\top)^\top \in \Gamma^{-i}$ to denote the strategy and 3strategy space of all agents except the agent $i$. It should be noted that open-loop strategies provide equivalence between strategy and control actions at all time instants ($\gamma \equiv \{u_t\}_{\{1,\ldots,T-1\}}$), which in turn determines the states of the whole system at all time steps given the initial state.

\vspace{-0.04cm}
When racing, each agent has their own objective of finishing the race before the opponents while satisfying some constraints such as being on the race track and avoiding collisions with other agents. As discussed before, let these constraints be denoted by $g(x_t,u_t,t) \leq 0$. The agents' objectives should incentivize behaviors such as moving as fast as possible along the race track and blocking the other agents if they are trying to overtake. Let the objective function of each agent be denoted by a function $J^i: \mathcal{X}\times\Gamma \rightarrow \mathbb{R}$ which is a function of the initial state of the game and the strategies of all the agents. Taking inspiration from \cite{9329208}, for each agent $i$, we choose $J^i$ to be of the form
\begin{equation}\label{eq:cost}
    J^i(x_0,\gamma) = -r(x^i_{T}) + \alpha \sum_{t=0}^{T-1} \sum_{\substack{j \in [N] \\ j \neq i}} d(x^i_t,x^j_t)^2,
\end{equation}
where $r(\cdot)$ is the path covered along the race track, and $\alpha$ is a hyperparameter that determines the amount of weight given to generating blocking behavior with respect to moving along the track as fast as possible. We choose the objective functions to be of the form in \eqref{eq:cost} with the following motivations:
\begin{itemize}
    \item For a fixed planning horizon $T$, each agent wants to maximize the path covered along the race track, which corresponds to minimizing $-r(x^i_T)$.
    \item We take motivation from the sensitivity terms introduced in \cite{9329208} to incentivize competitive behaviors. Each agent wants to prevent other agents from overtaking themselves and remain in close proximity to their opponents who are ahead of them. Our method incentivizes blocking behaviors if the agent using our method is ahead of other agents. This is because, while the first term in \eqref{eq:cost} encourages the agent to progress along the track, the second term encourages being in proximity to other agents, which generates blocking behavior if the agent using our method is ahead.
    On the other hand, if the opponents are ahead, minimizing the distance to opponents is equivalent to catching up. Furthermore, in such cases, due to the first term, our method will incentivize the agent not only to catch up but also to overtake.
    \item Note that the $\alpha$ term in \eqref{eq:cost} indicates the aggressiveness of agents, as a larger value of $\alpha$ will incentivize agents to be in proximity with other agents. It should also be noted that, due to the collision constraints, the distance between each pair of agents has a lower bound, which guarantees the safety of the generated trajectories.
\end{itemize}
When racing, each agent wants to win the race, which translates to minimizing its objective. However, each agent cannot minimize its objective and satisfy its constraints in isolation, as the value of its objective depends on the actions and states of other agents as well. Therefore, none of the agents can independently minimize their objectives as the objectives and constraints are inherently coupled. Consequently, to capture the dependence of agents upon one another, they must seek equilibria of the dynamic game underlying their interactions.

We denote such a non-cooperative general-sum constrained dynamic game in a compact form as $\mathcal{G} := \left([N], \{\Gamma_i\}_{i\in[N]}, \{J_i\}_{i\in[N]},\{\mathcal{C}_k\}_{k\in\{0,\ldots,T\}},f,x_0 \right)$. We model the outcome of this competitive racing scenario by open-loop constrained Nash equilibria, also known as generalized Nash equilibria, of the constrained dynamic game $\mathcal{G}$. 
\begin{definition}
A set of strategies $\gamma^*$ is a generalized Nash equilibrium of the game $\mathcal{G} = \left([N], \{\Gamma_i\}_{i\in[N]}, \{J_i\}_{i\in[N]},\{\mathcal{C}_k\}_{k\in\{0,\ldots,T\}},f,x_0 \right)$ if the following holds for each agent $i \in [N]$:
\begin{align}\label{eq:Nash-definition}
    & J^i(x_0,\gamma^{*}) \leq   J^i(x_0,{\gamma^{i}},{\gamma^{-i}}^*) \nonumber\\
    & \forall \left(x_t,{\gamma^{i}}(x_0,t),{\gamma^{-i}}^*(x_0,t)\right) \in \mathcal{C}_t, t \in \{0,\ldots,T-1\}, \nonumber \\
    & \forall \; x_{T} \in \mathcal{C}_{T}.
\end{align}
\end{definition}
This definition essentially means that no agent would want to deviate from their equilibrium strategy to any other feasible strategy as it would result in incurring a higher cost. It is important to notice that finding a solution to \eqref{eq:Nash-definition} will require solving a set of N-coupled constrained optimal control problems, which is challenging to solve tractably in practice. In the next section, we will discuss how we can efficiently find the equilibria of the game tractably in real time.

\section{Dynamic Potential Games}
Dynamic potential games are a class of games for which Nash equilibria can be found by solving a single optimal control problem instead of having to solve several coupled optimal control problems. If a game is a potential game, there exists a potential function $P$, and the Nash equilibria of the game can be computed by minimizing this potential function. This largely simplifies the computation of Nash equilibria. This property of potential games has been recently leveraged in the context of multi-agent navigation and has been shown to simplify the problem of trajectory planning significantly~\cite{kavuncu2021potential,williams2023distributed}. We formally define dynamic potential games as follows:

\begin{definition}\label{def:potential}
A non-cooperative constrained dynamic game $\mathcal{G} := \left([N], \{\Gamma_i\}_{i\in[N]}, \{J_i\}_{i\in[N]},\{\mathcal{C}_k\}_{k\in\{0,\ldots,T\}},f,x_0 \right)$ is a constrained potential dynamic game if there exists a potential function $P:\astates\times\Gamma \rightarrow \bR$ such that for every agent $i\in [N]$, and every pair of strategies $\gamma^i \in \Gamma^i,\nu^i \in \Gamma^i$, once we fix the set of strategies $\gamma^{-i} \in \Gamma^{-i}$, we have:
\begin{align}\label{eq:potential-condition}
    & J^i(x_0,\gamma^i,\gamma^{-i}) - J^i(x_0,\nu^i,\gamma^{-i}) \nonumber\\ &\quad = P(x_0,\gamma^i,\gamma^{-i}) - P(x_0,\nu^i,\gamma^{-i}).
\end{align}
\end{definition}

Definition~\ref{def:potential} essentially states that a game $\mathcal{G}$ is a constrained dynamic potential game if there exists a global potential function that captures the change in the cost of all agents when they change their strategy while other agents keep their strategies fixed.



\begin{proposition}
If a game is a constrained dynamic potential game, it can be shown that generalized open-loop Nash equilibria of the dynamic potential game $\mathcal{G}$ can be computed by solving the following multivariate optimal control problem: 
\begin{align}\label{mopc}
      \underset{\gamma \in \Gamma}{\text{minimize}} \quad & P(x_0,\gamma) \nonumber\\
     \text{subject to} \quad& x_{k+1} = f(x_k,u_k,k), \; x_0 \; \text{given} \nonumber \\
    & g(x_k,u_k,k) \leq 0.
\end{align}
\end{proposition}
\begin{proof}
Proof can be found in Theorem-1 from~\cite{bhatt2022efficient}.
\end{proof}

\section{RAPID: Scalable Planner for Racing}
In this section, we show that our proposed multi-agent racing game is an instance of a dynamic constrained potential game. We will prove that under cost structures~\eqref{eq:cost}, the resulting constrained dynamic game will be a dynamic potential game whose generalized Nash equilibria can be found by solving one single constrained optimal control problem. The following theorem characterizes the cost structures under which the game $\mathcal{G}$ will be a constrained dynamic potential game.

\begin{theorem}\label{thm:1}
    A game $\mathcal{G}$ where every agent optimizes \eqref{eq:cost} with constraints given in \eqref{eq:track_constraint}, \eqref{eq:collision_constraint} and \eqref{eq:control_constraint}, is a constrained dynamic potential game with the potential function
    \begin{equation}\label{eq:potential-function}
        P(x_0,\gamma) = -\sum_{i=1}^{N}r(x^i_{T}) + \alpha \sum_{t=0}^{T-1} \sum_{i=1}^{N}\sum_{j=i+1}^{N} d(x^i_t, x^j_t)^2.
    \end{equation}
     A generalized Nash equilibrium of $\mathcal{G}$ can be computed by solving the following single optimal control problem
\begin{align}\label{eq:potential_game}
      \underset{\gamma \in \Gamma}{\text{minimize}} \quad & P(x_0, \gamma)\\
     \text{subject to} \quad& u_{min} \leq u^i_t  \leq u_{max}, \quad \forall i \in [N], \forall t \in [T-1]; \nonumber\\
     & d(x^i_t, x^j_t) > d_{min}, \quad \forall i \neq j \in [N], \forall t \in \left[T\right]; \nonumber\\
     & h(x^i_t) \leq 0,  \quad \forall i \in [N], \forall t \in \left[T\right].
\end{align}
\end{theorem}

\begin{proof}
    See Appendix \ref{appdix:pf1}.
\end{proof}

Note that Theorem~\ref{thm:1} indicates that to find the equilibria of the games, it suffices to solve~\eqref{eq:potential_game}. Since~\eqref{eq:potential_game} is a single constrained optimal control problem, one can use any existing constrained trajectory optimizer to solve it. 

Our overall algorithm is described in Algorithm~\ref{alg:ego}, where ego represents the agent that uses our algorithm. 

Note that while the $\alpha$ term in Algorithm~\ref{alg:ego} stays the same during each planning horizon, its value can be different between planning horizons. Since the $\alpha$ value indicates the aggressiveness of agents, we introduce the idea of \emph{adaptive aggressiveness scaling}, which scales $\alpha$ depending on the average distance of ego to other agents. 

In line 1 of Algorithm~\ref{alg:ego}, we sum over the squared distance of the ego agent to other agents. When the average squared distance of ego to other agents is above some threshold $D$, we use a smaller $\alpha$ value ($\alpha_{inactive}$ in line 2). Otherwise, a larger value ($\alpha_{active}$ in line 4) is used when ego is in proximity with other agents. When ego is far from others, the sense of being ``aggressive" is not of relevance and thus the aggressiveness term is inactive. When ego is in proximity with other agents, having a larger $\alpha$ value indicates the aggressiveness term is active and ego will try to catch up with the agent ahead or block the agent behind, depending on the relative positions, which generates overtaking and blocking maneuvers. 
\vspace{-0.2cm}
\begin{algorithm}
\caption{RAPID with Adaptive Aggressiveness Scaling}\label{alg:ego}
\hspace*{\algorithmicindent} \textbf{Input:}
all agents' current states $\{{x}^{i}\}_{i \in [N]}$, active distance threshold $D$, alpha inactive $\alpha_{inactive}$, alpha active $\alpha_{active}$.
\begin{algorithmic}[1]
\If {$\sum_{i\in [N]}\|d({x}^{ego}, {x}^{i}) \|^2 > (N-1)D$} 
	\State  $\alpha = \alpha_{inactive}$
\Else
	\State $\alpha = \alpha_{active}$
\EndIf
\State solve \eqref{eq:potential_game} with $\alpha$.
\end{algorithmic}
\end{algorithm}

\vspace{-0.4cm}
\section{Simulation Results}
In this section, we demonstrate the performance of our algorithm by considering two agents racing against each other.

We use a discrete variation of the Dubins car model for each agent. The state vector of each agent $i \in [N]$ at time step $t$ is:
\vspace{-0.2cm}
\begin{align}
    x^i_t = \left[p^i_{x,t}, p^i_{y,t}, v^i_t, \theta^i_t\right]
\end{align}
where $p^i_{x,t}, p^i_{y,t}$ represent the $x, y$ coordinates of the agent in the world frame and $v^i_t, \theta^i_t$ represent its speed and orientation. 
Each agent's dynamics is given by:
\begin{align*}
    {x}^i_{t+1} = 
    \begin{bmatrix}
        p^i_{x,t} + v^i_t\text{cos}(\theta^i_t)\delta{t}\\
        p^i_{y,t} + v^i_t\text{sin}(\theta^i_t)\delta{t}\\
        v^i_t + a^i_t \delta{t}\\
        \theta^i_t + \omega^i_t \delta{t}
    \end{bmatrix}
\end{align*}
where $a^i_k, \omega^i_t$ are the acceleration and angular velocity (control inputs) of agent $i$ at time step $t$, and $\delta{t}$ represents the time step length.
During the start of each planning horizon, all agents have access to the ground truth of their own current states as well as the current states of other agents.
For simplicity, we call the agent that uses our algorithm the {\it ego}. We compare our algorithm with the racing algorithm developed in~\cite{9329208}, which will be called the {\it baseline}. The baseline algorithm used iterative best responses and was shown to be able to generate competitive behaviors.

We solve both the ego's optimization problem and the baseline's optimization problem with do-mpc
\cite{LUCIA201751}, which models the problems symbolically with CasADi \cite{Andersson2019} and solves them with IPOPT. \cite{Wachter:2006wt}. We choose the planning horizon to be 5 time steps for each agent, where each time step corresponds to $0.1s$.

The race track that we used is shown in Fig \ref{customized_track}.
\begin{figure}
    \centering
    \includegraphics[scale=0.4]{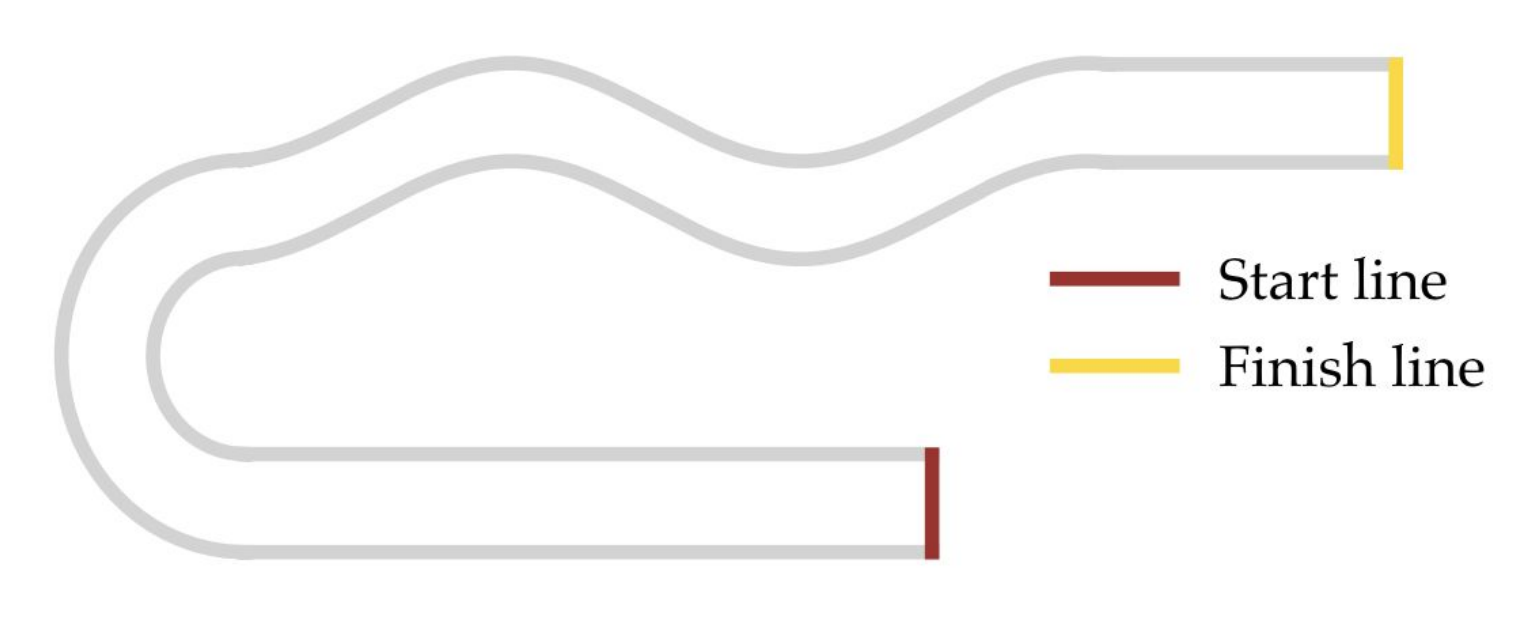}
    \caption{The top view of the customized track. The brown line is the starting line, and the gold line is the finish line.}
    \label{customized_track}
    \vspace{-0.6cm}
\end{figure}
The start and finish lines are marked with brown and gold, respectively. 
We run the experiments between the ego and the baseline for 50 random initial positions.
The initial positions are randomly sampled within a fixed area around the starting line so that the initial distance between the two agents is within $\left[1.0, 1.5\right] m$ . The agent with a leading starting position is given a maximum speed of $2.4 ~m/s$, and the other agent is given a maximum speed of $2.5 ~m/s$. Note that we chose this setup to ensure enough interactions between the two agents.

The results are demonstrated in Table \ref{hard_track_sim_stats}.
Note that the solve time here represents the time that it takes the agent to solve its action for one time step.
As shown in Table~\ref{hard_track_sim_stats}, our algorithm has a significant advantage over the baseline in terms of both the solve time and the quality of the generated trajectories. Our algorithm is about 4 times faster than the baseline while winning the race more frequently. We find that the ego is able to push the baseline against the outer boundary of the track during the first corner, and the baseline would be stuck at the boundary until the ego passes. This blocking behavior is plotted in Fig \ref{hard_track_push}. 

\begin{figure}
    \centering
    \includegraphics[scale=0.36]{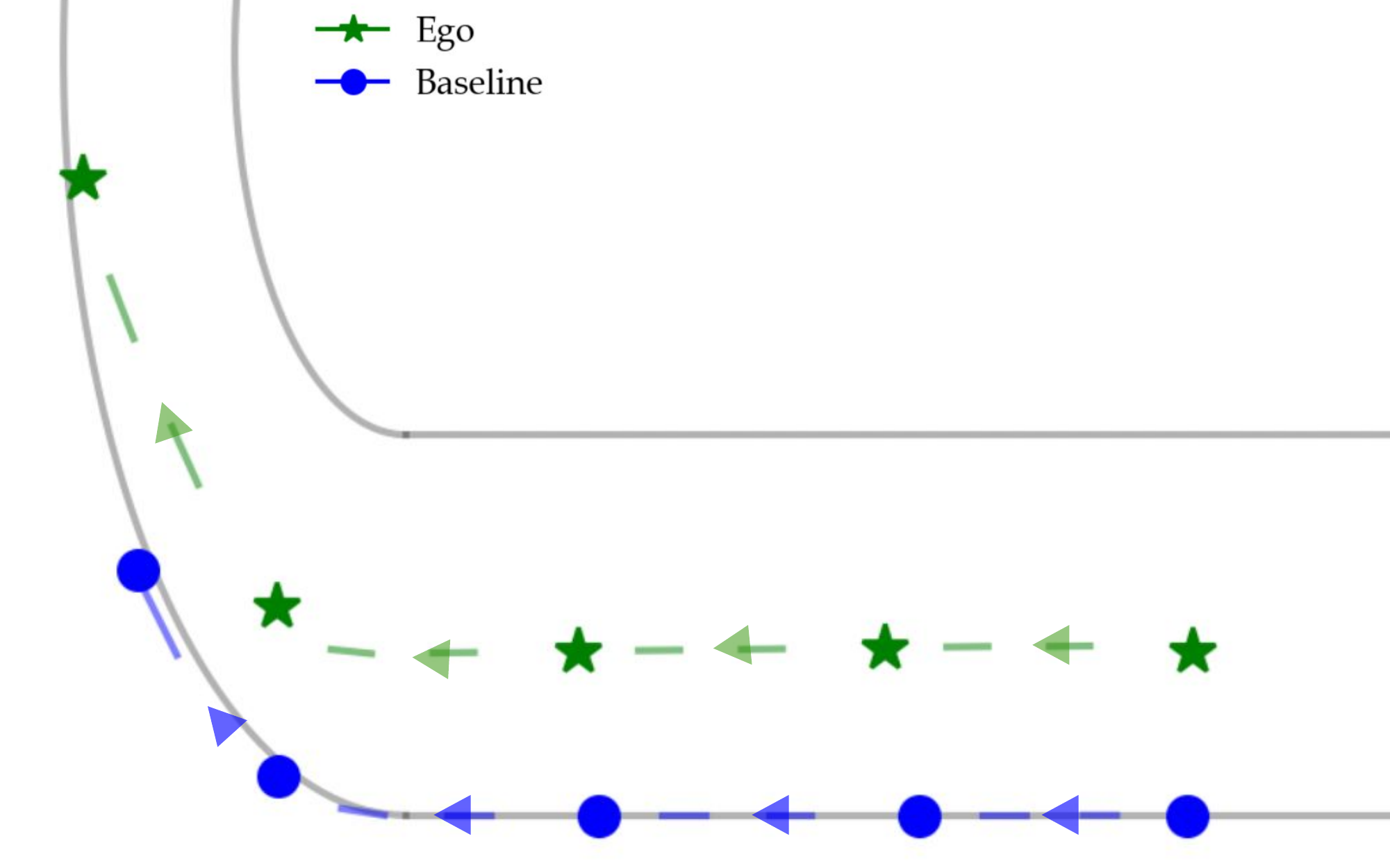}
    \caption{The ego (plotted with green color) pushes the baseline (plotted with blue color) against the boundary.}
    \label{hard_track_push}
\end{figure}

\begin{table}[h]
\begin{center}
\begin{tabular}{|c|c|c|}
\hline
 & {\bf Average Solve Time (s)} & {\bf Win} \\
\hline
  {\bf Ego} & $0.051 \pm 0.010$  & 36 \\
 \hline
  {\bf Baseline} & $0.215 \pm 0.084$ & 14 \\ 
\hline
\end{tabular}
\end{center}
\caption{Simulation Results for the Customized Track}
\label{hard_track_sim_stats}
\end{table}

To test the scalability of our approach, we ran our algorithm and the baseline algorithm against the reactive MPC algorithm, respectively, where the reactive MPC algorithm is described in Equation \eqref{standard_optm_problem}. We ran simulations for various numbers of agents, where one of them used our algorithm or the baseline algorithm, while the others used the reactive MPC algorithm. The results are recorded in Table \ref{solve_time_diff_num_agents}. 
As indicated by the results, our algorithm is able to scale better than the baseline.


\begin{table}[h]
\begin{center}
\begin{tabular}{|p{1.2cm}|c|c|c|}
\hline
 {\bf Number of Agents} & {\bf 2} & {\bf 5}  & {\bf 10}\\
\hline
  {\bf RAPID} & $\mathbf{0.078 \pm 0.046}$ &  $\mathbf{0.229 \pm 0.535}$ & $\mathbf{0.548 \pm 0.400}$ \\
\hline
  {\bf Baseline} & $0.217 \pm 0.061$ & $0.637 \pm 0.110$  &  $1.456 \pm 0.187$\\
  \hline
\end{tabular}
\end{center}
\caption{Mean Solve Time (in seconds) Comparisons for Different Numbers of Agents.}
\label{solve_time_diff_num_agents}
\vspace{0.2cm}
\end{table}



\vspace{-0.3cm}
\section{Hardware Experiments}
To further evaluate the real-time capabilities of our framework, we set up experiments in hardware using Crazyflie
2.0 quadcopters within the Robot Operating System (ROS) framework~\cite{quigley2009ros}. We used a
Vicon motion capture system to collect the state information of the agents. We also used Crazyswarm \cite{crazyswarm} to send waypoint commands to the Crazyflies. We then used a low-level controller to follow the waypoints.

The state vector of each quadrotor $i \in \left[N\right]$, at each time step $t$, is:
\begin{align}
    x^i_t &= [p^i_{x, t}, ~p^i_{y, t}, ~p^i_{z, t}, ~\phi^i_{t}, ~\theta^i_t, ~\psi^i_t],
\end{align}
which consists of its position ($[p^i_{x, t}, ~p^i_{y, t}, ~p^i_{z, t}]$) and orientation ($[\phi^i_{t}, ~\theta^i_t, ~\psi^i_t]$). We model the quadrotor dynamics as given by $\dot{x}^i_t = u^i_t$, where $u^i_t = [u^i_{1, t}, ~u^i_{2, t}, ~u^i_{3, t}, ~u^i_{4, t}, ~u^i_{5, t}, ~u^i_{6, t}]$ is the control input. The first three entries of $u^i_k$ correspond to the linear velocities and the last three entries correspond to the angular velocities. 
We use this simplified model since we can send waypoints directly as commands to the quadcopters through the Crazyswarm framework \cite{crazyswarm}, and an in-place low-level controller will be used to track the waypoints.

We also implemented a model predictive controller (MPC) for comparison purposes. This MPC planner is in fact a reactive planner that treats other agents as obstacles that need to be avoided. We call this planner {\it reactive MPC}.
When running reactive MPC, each agent aims to maximize its progress along the track while satisfying input constraints, collision constraints, and track constraints. The reactive MPC predicts the position of other agents by assuming other agents move with constant velocity during each planning horizon, where the constant velocity comes from finite difference estimation. For each agent $l \in [N]$, the reactive MPC solves the following optimization problem: 
\begin{align}\label{standard_optm_problem}
      \underset{\gamma^{l} \in \Gamma^{l}}{\text{minimize}} \quad & -r(x^{l}_{T},T)\nonumber\\
     \text{subject to} \quad& u_{min} \leq \|u^{l}_t\|  \leq u_{max}, \forall i \in [N], \forall t \in \left[T\right]; \nonumber\\
     & d(x^{l}_t, x^j_t) > d_{min}, \forall j \in [N]\setminus \{l\}, \forall t \in \left[T\right]; \nonumber\\
     & h(x^{l}_t) \leq 0,  \quad \forall t \in \left[T\right]
\end{align}
\vspace{-0.25cm}
\subsection{Two-agent Race}
We begin by considering a race between two agents, where one agent uses our algorithm while the other uses the reactive MPC algorithm.
The track is shown in Fig \ref{fig:two_player_hardware_experiment}. The initial positions are chosen with some randomness. 
The quadrotor that starts ahead is given a maximum speed of $1.5$ m/s, and the one that starts behind is given a maximum speed of $1.8$ m/s.

The results are displayed in Table \ref{two_player_hardware_stats}, where {\it Overtakings} records the number of times where the ego agents started behind but reached the finish line first and {\it Defendings} records the number of times where it started ahead and reached the finish line first. 
As shown by the results, the ego agent significantly outperforms the reactive MPC by winning most of the races.
Fig \ref{fig:two_player_hardware_experiment} shows both agents' trajectories in two of the hardware experiments.
As we can see, the ego quadrotor is able to overtake the opponent (the quadrotor that uses reactive MPC) when it starts behind and then block the opponent when the opponent has a possibility to overtake. In both cases, the ego is able to secure a leading position without causing any collision.
This also demonstrates that our algorithm is able to exploit the trade-off between making progress and blocking opponents without causing constraint violations.

\begin{table}
\begin{center}
\begin{tabular}{|c||c|}
\hline
{\bf Overtakings} & {\bf Defendings}\\
\hline
10/10 & 8/10\\
\hline
\end{tabular}
\end{center}
\vspace{-0.3cm}
\caption{Results for Two-agent Hardware Experiments. Our method was able to overtake 10 times out of the 10 trials when it started behind. And it was able to defend its leading position 8 times out of the 10 trials when it started ahead.}
\label{two_player_hardware_stats}
\vspace{-0.6cm}
\end{table}

\vspace{-0.1cm}
\subsection{Three-agent Race}

We also test our algorithm with a more complicated setting where there are three agents and the track becomes narrower. One of the agents uses our algorithm, and the other two use the reactive MPC algorithm.

The track is similar to the one used in the two-agent race but is made to be narrower in order to test our algorithm in a more competitive scenario.
The initial conditions are selected such that the quadrotors engage in interactions and have enough space for maneuvering. The speed limit for the leading quadrotor was set to be $1.5 m/s$, the middle quadrotor $1.65 m/s$, and the quadrotor behind had a speed limit of $1.8 m/s$.

We performed the experiments 15 times where the ego quadrotor starts first 5 times, starts second 5 times, and starts third 5 times. 
The results are recorded in table \ref{three_player_hardware_stats}. As we can see, the ego was able to win most of the races, even when it starts behind or starts ahead but given a lower speed limit. 
We observe that our algorithm is able to take advantage of the collision avoidance constraints of other agents and generate overtaking or blocking, and it is doing so without causing any collision.
Fig \ref{fig:three_player_experiment} shows an instance of such a blocking maneuver during a three-agent race. As we can see, the ego (represented by green color)
intentionally moves toward agent 1 (represented by blue color) to block it. As a result, the ego is able to make more progress than agent 1 despite being given a lower speed limit. Note that the experimental snapshots are not taken from a top view, so the relative position of agents may look different from the top view plots.
For more details on the hardware experiments of three drones, we recommend that interested readers check out the videos.

\begin{table}[h]
\centering
\begin{center}
\begin{tabular}{|c||c||c||c|}
\hline
 & {\bf Started 1st} & {\bf Started 2nd} & {\bf Started 3rd}\\
 \hline
{\bf Ego Finished 1st} & 3/5 & 4/5 & 4/5\\
\hline
\end{tabular}
\end{center}
\vspace{-0.2cm}
\caption{Results for Three-agents Hardware Experiments}
\label{three_player_hardware_stats}
\end{table}

\vspace{-0.3cm}
\section{Conclusion and Future Work}
\vspace{-0cm}
In this work, we presented RAPID, an efficient algorithm for trajectory planning in autonomous racing. Our method exploits the special cost structure under which the race can be formulated as a potential game, which enables us to obtain a generalized Nash equilibrium of the game by solving a single constrained optimal control problem. 
We demonstrated the performance of the algorithm through both simulation studies and hardware experiments. We showed that RAPID could generate delicate interactive maneuvers such as overtaking and blocking while avoiding collisions with other agents. 
Currently, some parameters such as the active distance threshold and $\alpha_{inactive}, \alpha_{active}$ in Algorithm-\ref{alg:ego} used in RAPID need to be determined empirically. In future work, we wish to provide theoretical results on the effects of these parameters. 









%
%
%

\bibliography{IEEEabrv,bibliography}
\bibliographystyle{IEEEtran}

\appendix
\subsection{Proof of Theorem \ref{thm:1}}\label{appdix:pf1}
\begin{proof}
Since agents' dynamics are separable, the state and control input of each agent can be written as a function of their strategies. We will explicitly denote this dependence in this proof. Consider two sets of open-loop strategies $\gamma := (\gamma^{i^{\top}},\gamma^{{-i}^\top})^\top$ and  $\gamma^\prime := (\nu^{i^{\top}},\gamma^{{-i}^\top})^\top$. Let the states and control inputs generated by the set of strategies $\gamma$ be $\{x_t\}_{ t\in \{0,1,\ldots,T\}}$ and $\{u_t\}_{ t\in \{0,1,\ldots,T-1\}}$ respectively. Likewise, let the states and control inputs generated by the set of open-loop strategies $\gamma^\prime$ be $\{x^\prime_t\}_{ t\in \{0,1,\ldots,T\}}$ and $\{u^\prime_t\}_{ t\in \{0,1,\ldots,T-1\}}$ respectively. Note that since between $\gamma$ and $\gamma^\prime$, only the strategy of agent $i$ is being changed, all the other agents' states and actions will be the same when using strategies $\gamma$ and $\gamma^\prime$. Therefore, we have the following relation:
\begin{equation}\label{eq:relation}
        {x^\prime}^j_t = x^j_t, {u^\prime}^j_t = u^j_t \quad \forall j \neq i, j \in [N].
    \end{equation}

Now, we can compute the difference of the cost for agent $i$ as
    \begin{align}
        & J^i(x_0,\gamma) - J^i(x_0,\gamma^\prime) = -r(x^i_T) + r({x_T^\prime}^i) \nonumber \\
        & + \alpha \sum_{t=0}^{T-1} \sum_{\substack{j \in [N] \\ j \neq i}} \left(d(x^i_t,x^j_t)^2 - d(x^{\prime^i}_t , {x^\prime}^j_t)^2\right).
    \end{align}
Using results from \eqref{eq:relation}, since we have that ${x^\prime}^j_t = x^j_t, {u^\prime}^j_t = u^j_t$, for all $j \neq i$, we obtain
\begin{align}
        & J^i(x_0,\gamma) - J^i(x_0,\gamma^\prime) = -r(x^i_T) + r(x^{\prime^i}_T) \nonumber \\
        & \quad + \alpha \sum_{t=0}^{T-1} \sum_{\substack{j \in [N] \\ j \neq i}} \left(d(x^i_t, x^j_t)^2 - d(x^{\prime^i}_t, {x}^j_t)^2\right).
\end{align}

Now, let's consider the change in the potential function~\eqref{eq:potential-function} when the strategy of agent $i$ changes from $\gamma^i$ to $\nu^i$:
\begin{align}\label{eq:change_potential}
    & P(x_0,\gamma) - P(x_0,\gamma^\prime) = -\sum_{k=1}^{N}\left(r(x^k_T) - r(x^{\prime^k}_T)\right) \nonumber \\
    &  \quad + \alpha \sum_{t=0}^{T-1} \sum_{k=1}^{N}\sum_{l=k+1}^{N} \left(d(x^k_t, x^l_t)^2 - d({x^\prime}^k_t, {x^\prime}^l_t)^2\right).
\end{align}
Using \eqref{eq:relation}, we have:
\begin{align}\label{eq:first_term}
    & \sum_{k=1}^{N}\left(r(x^k_T) - r({x^\prime}^k_T )\right) = r(x^i_T) - r(x^{\prime^i}_T),
\end{align}
and similarly, using~\eqref{eq:relation},
\begin{align}\label{eq:second_term}
    & \sum_{k=1}^{N}\sum_{l=k+1}^{N} \left(d(x^k_t, x^l_t)^2 - d({x^\prime}^k_t, {x^\prime}^l_t)^2\right) \nonumber \\
    & \quad = \sum_{k=1}^{i}\sum_{l=k+1}^{N}\left(d(x^k_t, x^l_t)^2 - d({x^\prime}^k_t, {x^\prime}^l_t)^2\right) \nonumber \\
    & \quad + \sum_{k=i+1}^{N}\sum_{l=k+1}^{N}\left(d(x^k_t, x^l_t)^2 - d({x}^k_t, {x}^l_t)^2\right) 
\end{align}
As we can see, second term in \eqref{eq:second_term} is zero. In the first term, for each value of $0<k<i$, all the terms corresponding to all the values of $l$ except $l=i$ will be zero due to \eqref{eq:relation}. Therefore,
\begin{align}\label{eq:third_term}
    & \sum_{k=1}^{N}\sum_{l=k+1}^{N} \left(d(x^k_t, x^l_t)^2 - d({x^\prime}^k_t, {x^\prime}^l_t)^2\right) \nonumber \\
    & \quad = \sum_{k=1}^{i-1}\left(d(x^k_t, x^i_t)^2 - d({x}^k_t, {x^\prime}^i_t)^2\right) \nonumber \\
    & \quad + \sum_{l=i+1}^{N}\left(d(x^i_t, x^l_t)^2 - d({x^\prime}^i_t, {x}^l_t)^2\right) \nonumber \\
     & \quad = \sum_{\substack{j \in [N] \\ j \neq i}} \left(d(x^i_t, x^j_t)^2 - d(x^{\prime^i}_t, {x}^j_t)^2\right)
\end{align}
Putting together~\eqref{eq:first_term} and \eqref{eq:third_term} in \eqref{eq:change_potential}, we obtain

\begin{align}
        & P(x_0,\gamma) - P(x_0,\gamma^\prime) = -r(x^i_T) + r(x^{\prime^i}_T) \nonumber \\
        & \quad + \alpha \sum_{t=0}^{T-1} \sum_{\substack{j \in [N] \\ j \neq i}} \left(d(x^i_t, x^j_t)^2 - d(x^{\prime^i}_t, {x}^j_t)^2\right). \nonumber \\
        \Rightarrow & P(x_0,\gamma) - P(x_0,\gamma^\prime) = J^i(x_0,\gamma) - J^i(x_0,\gamma^\prime), \forall i \in [N].
\end{align}

Therefore, $P$ is indeed the potential function for the underlying dynamic game using Definition~\ref{def:potential}. Therefore using results from Theorem 1 in \cite{bhatt2022efficient}, we can compute generalized open-loop Nash equilibria of the original game by solving the multivariate optimal control problem given in \eqref{eq:potential_game}.

\end{proof}

\end{document}